\documentclass[runningheads]{llncs}

\usepackage[T1]{fontenc}
\usepackage{graphicx}
\usepackage{natbib}
\usepackage{amsmath}
\usepackage{amsfonts}

\usepackage{algorithm}
\usepackage[noend]{algpseudocode}
\usepackage{hyperref}
\usepackage{todonotes}
\usepackage{pgfplots}
\usepgfplotslibrary{groupplots}
\usepackage{tikz}
\usetikzlibrary{automata, positioning}

\usepackage{listings}
\usepackage{xcolor}  

\begin{document}

\title{Passive Model Learning of Visibly \\ Deterministic Context-free Grammars}

\author{Edi Mu\v{s}kardin\inst{1}\orcidID{0000-0001-8089-5024} \\ \and Tamim Burgstaller\inst{2}\orcidID{0009-0007-4522-8497}}

\authorrunning{Mu\v{s}kardin and Burgstaller}

\institute{Silicon Austria Labs \\
\email{edi.muskardin@silicon-austria.com}
\and 
Graz University of Technology \\
\email{tamim.burgstaller@tugraz.at}}

\maketitle              

\begin{abstract}
We present PAPNI, a passive automata learning algorithm capable of learning deterministic context-free grammars, which are modeled with visibly deterministic pushdown automata. PAPNI is a generalization of RPNI, a passive automata learning algorithm capable of learning regular languages from positive and negative samples. PAPNI uses RPNI as its underlying learning algorithm while assuming a priori knowledge of the visibly deterministic input alphabet, that is, the alphabet decomposition into symbols that push to the stack, pop from the stack, or do not affect the stack.

In this paper, we show how passive learning of deterministic pushdown automata can be viewed as a preprocessing step of standard RPNI implementations. We evaluate the proposed approach on various deterministic context-free grammars found in the literature and compare the predictive accuracy of learned models with RPNI.

\end{abstract}

\section{Introduction}

Automata learning, also known as model learning~\cite{DBLP:conf/dagstuhl/AichernigMMTT16}, is a field of theoretical computer science that concerns itself with automatic identification and modeling of a system's structure and behavior from observations. Learned models are encoded as finite-state machines, most often as DFAs, Mealy machines, and even Markov Decision Processes~\cite{DBLP:conf/fm/TapplerA0EL19}. Nowadays automata learning is often used in automated modeling, reverse engineering, and analysis of complex steteful systems, such as network protocols~\cite{DBLP:conf/icst/AichernigMP21,DBLP:journals/fmsd/PferscherA22,DBLP:conf/fm/PferscherWAM23}, software based on finite-state machines~\cite{DBLP:conf/spin/Ganty24}, and even of machine learning components~\cite{DBLP:conf/cdmake/MayrY18,DBLP:journals/ml/WeissGY24,DBLP:conf/icst/TapplerMAK24}.

Automata learning can be divided into two main categories. \textit{Active} automata learning learns the behavior of the system under learning (SUL) by actively interacting with the SUL, while \textit{passive} automata learning learns solely from provided data, e.g., in the form of system traces. Notable active automata learning algorithms are $L^*$~\cite{DBLP:journals/iandc/Angluin87}, $KV$~\cite{DBLP:journals/cacm/Valiant84}, and $TTT$~\cite{DBLP:conf/rv/IsbernerHS14}, while RPNI~\cite{rpni} and EDSM~\cite{DBLP:conf/icgi/LangPP98} and popular choices in the passive learning paradigm. While these algorithms differ in their approach to automata learning, they all encode the SUL's behavior as a regular language, that is, with automata such as DFAs and Mealy machines. 

However, the behavior of many real-world systems cannot be captured by regular languages alone. For instance, consider the language of balanced parentheses, e.g., (), (()), or (()(())). Recognizing such a language requires matching each opening parenthesis with a corresponding closing one, which involves keeping track of nested dependencies. Finite automata cannot deal with such stackful behavior; attempting to model this as a regular language would lead to an infinite number of states (due to ever-increasing parentheses nesting depth). In contrast, context-free grammars~\cite{DBLP:books/daglib/0016921} allow us to represent such stack-based behavior, which we can then encode with pushdown automata. In this paper, we show how RPNI can be extended to handle stackful behavior, with the assumption that we know the decomposition of the input alphabet into internal, call, and return symbols.

This paper is structured as follows: in Section~\ref{sec:related_work}, we outline related work, mostly focusing on existing passive automata learning methods. Section~\ref{sec:prelim} presents necessary preliminaries and definitions. In the Section~\ref{sec:method}, we present the algorithm behind PAPNI, and we reason about its correctness in Section~\ref{sec:proof}. We evaluate PAPNI in Section~\ref{sec:evaluation} and discuss possible limitations and their mitigations of PAPNI in Section~\ref{sec:limitations}. Finally, we conclude the paper and outline possible future work in Section~\ref{sec:conclusion}.

\section{Related Work}\label{sec:related_work}

In this section, we briefly outline the passive automata learning landscape, as well as the existing method of active inference of deterministic context-free grammars. 

The most notable passive automata learning algorithm for learning deterministic systems is Regular Positive and Negative Inference (RPNI)~\cite{rpni,10.5555/1830440}. The name of the algorithm reflects its origin as a DFA inference method, given that DFAs distinguish between positive and negative sequences. RPNI was later extended to learn Mealy and Moore machines~\cite{DBLP:journals/sttt/GiantamidisTB21,DBLP:conf/cav/IsbernerHS15,DBLP:journals/isse/MuskardinAPPT22}, enabling it to model formalisms with more expressive output domains. RPNI works as follows: 
it starts by spanning all input sequences and their corresponding label into a prefix tree acceptor (PTA). That is, for each input sequence a new path is added to the PTA, such that the inputs of a sequence lead to a node that is labeled with the corresponding label. 
This tree initially captures all observed transitions without any generalization. RPNI then processes all PTA nodes in a breadth-first manner and merges two nodes if they are deemed compatible. To check state compatibility, RPNI checks whether the merger of two states remains consistent with both positive and negative examples and the resulting model remains deterministic. If these two properties are satisfied, the merger of states is deemed as valid, leading to a gradual refinement of the PTA while ensuring it accurately represents the input data.
=

Evidence-Driven State Merging (EDSM)~\cite{DBLP:conf/icgi/LangPP98} is an RPNI variation that implements a stricter state merging policy. That is, it does not greatly merge states with identical futures but implements a notion of evidence. EDSM merges only those states that display high statistical similarity (as indicated in the computed evidence score), which in turn improves the generalization capabilities of the learned model. EDSM's learning performance was demonstrated in the Abbadingo DFA learning competition, where it won first place. Note that EDSM can be used as a backend algorithm for PAPNI, as it has the same prerequisites as classic RPNI.

The $k$-tails algorithm~\cite{DBLP:journals/tc/BiermannF72} is similar to the classic RPNI approach but differs in its state merging criteria. Instead of considering the entire future behavior as RPNI does, $k$-tails merges states that exhibit identical behavior within a bounded length $k$. This makes $k$-tails particularly useful in scenarios where long-term behavior is less relevant or where computational efficiency is a priority, such as in software model inference and reverse engineering. Since $k$-tails only checks the future compatibility of states up to the bound $k$, the inferred model might display non-deterministic behavior.

So far we have considered the passive inference of deterministic models. Two notable algorithms for the passive inference of stochastic systems are  \textsc{Alergia}~\cite{DBLP:conf/icgi/CarrascoO94} and its extension \textsc{IoAlergia}~\cite{DBLP:journals/ml/MaoCJNLN16}. \textsc{Alergia} can be used to learn Markov chains, while \textsc{IoAlergia} infers Markov Decision Processes from the provided stochastic data. Unlike RPNI, these algorithms do not check the future output equivalence of states but rather their statistical similarity. That is, two states that are merged in their future/continuations display similar statistical output distributions (usually determined by a statistical compatibility check based on a Hoeffding bound~\cite{Hoeffding}).

Until this point we have discussed various passive learning approaches, none of which are able to model deterministic context-free grammars. The most notable algorithm capable of learning visibly deterministic context-free grammars is presented in the 6th chapter of Malte Isberner's PhD thesis~\cite{Malte_phdthesis}. This algorithm is implemented in LearnLib~\cite{DBLP:conf/cav/IsbernerHS15} and its variation is implemented in \textsc{AALpy}~\cite{DBLP:journals/isse/MuskardinAPPT22}. An extension of this algorithm was even used in JSON schema validation~\cite{DBLP:conf/tacas/BruyerePS23}. However, this algorithm differs from PAPNI in two important ways: (1) the presented algorithm falls into the active automata learning paradigm, and (2) it encodes the behavior of deterministic context-free grammars as 1-SEVPA~\cite{DBLP:conf/icalp/AlurKMV05}. The different formalism was chosen due to its favorable theoretical properties (in the context of active automata learning). We do not compare with their proposed algorithm, mostly due to the difference in the paradigm: in active learning, convergence to the ground truth is determined mostly by the strength of the equivalence oracle~\cite{zulu}, while in the passive learning setting, the quality of the learned model depends largely on the provided samples.

\newpage
\section{Preliminaries}\label{sec:prelim}

We start by introducing deterministic finite automata (DFA), as they serve as an underlying formalism in the PAPNI's backend. We then introduce visibly deterministic pushdown automata (VDPA)~\cite{DBLP:conf/stoc/AlurM04}, which can be used to model deterministic context-free grammars. The term ``visibly'' in visibly deterministic pushdown automata refers to the fact that the input alphabet of such models is divided into distinct categories: call symbols, return symbols, and internal symbols. In the remainder of this paper, we use the terms ``call'' and ``push'' and ``return'' and ``pop'' interchangeably, since call symbols push to the stack, and return symbols pop from the stack.

\begin{definition}[Deterministic Finite Automata (DFA)]
A DFA is a 5-tuple $\langle Q, I, \delta, q_0, F \rangle$, where:

\begin{itemize}
    \item $Q$ is a finite set of states,
    \item $I$ is a finite set of input symbols (input alphabet),
    \item $\delta: Q \times I \rightarrow Q$ is the transition function,
    \item $q_0 \in Q$ is the initial state, and
    \item $F \subseteq Q$ is the set of accepting (final) states.
\end{itemize}

\end{definition}

DFAs are a class of formal models that can be described by regular expressions. These languages are called regular languages and are found at the lowest level of the Chomsky hierarchy~\cite{DBLP:books/daglib/0016921}. 

A word (sequence of inputs) is processed by a DFA as follows. The automaton begins in the initial state $q_0$ and processes an input word one symbol at a time. At each step, the transition function $\delta$ determines the next state based on the current state and the current input symbol. This process continues until all symbols in the word have been read. If the automaton ends in a state that belongs to $F$, the word is accepted; otherwise, it is rejected. For a formal definition of various DFA properties, we point the reader to~\cite{DBLP:books/daglib/0016921}.  

\begin{definition}[Visibly Deterministic Pushdown Automaton (VDPA)]
A \emph{visibly deterministic pushdown automaton} is a tuple
\[
M = (Q, \Sigma, \Gamma, \delta, q_0, Z_0, F)
\]
where:
\begin{itemize}
    \item $Q$ is a finite set of states,
    \item $\Sigma$ is a finite input alphabet partitioned into three disjoint sets:
    \[
    \Sigma = \Sigma_{\mathit{call}} \ \dot{\cup} \ \Sigma_{\mathit{return}} \ \dot{\cup} \ \Sigma_{\mathit{internal}},
    \]
    where:
    \begin{itemize}
        \item $\Sigma_{\mathit{call}}$: symbols that \emph{push} onto the stack,
        \item $\Sigma_{\mathit{return}}$: symbols that \emph{pop} from the stack,
        \item $\Sigma_{\mathit{internal}}$: symbols that \emph{do not modify} the stack.
    \end{itemize}
    \item $\Gamma$ is a finite stack alphabet, and in the context of our paper $\Gamma = \Sigma_{\mathit{call}} \cup \{Z_0\}$, where $Z_0$ is a symbol denoting the empty stack,
    \item $\delta$ is a \emph{deterministic} transition function defined as the disjoint union of three partial functions:
    \begin{align*}
        \delta_{\mathit{call}} & : Q \times \Sigma_{\mathit{call}} \times \Gamma \to Q \times \Gamma \\
        \delta_{\mathit{return}} & : Q \times \Sigma_{\mathit{return}} \times \Gamma \to Q \times \Gamma^* \\
        \delta_{\mathit{internal}} & : Q \times \Sigma_{\mathit{internal}} \times \Gamma \to Q \times \Gamma
    \end{align*}
    such that:
    \begin{itemize}
        \item If $a \in \Sigma_{\mathit{call}}$, a symbol is pushed onto the stack
        \item If $a \in \Sigma_{\mathit{return}}$, then the top-of-stack symbol is removed,
        \item If $a \in \Sigma_{\mathit{internal}}$, then the stack is unchanged.
    \end{itemize}
    \item $q_0 \in Q$ is the initial state,
    \item $F \subseteq Q$ is the set of accepting states.
\end{itemize}
\end{definition}
A word $w \in \Sigma^*$ is \emph{accepted} by $M$ if, starting in $(q_0, Z_0)$ and processing $w$ according to $\delta$, the automaton ends in a configuration $(q, \epsilon)$ where $q \in F$ and the stack is empty. More general definitions of VDPA allow a non-empty stack for non-acceptance, but we restrict their acceptance to an empty stack, in the same fashion as~\cite{Malte_phdthesis}.
Note that $\Sigma_\textit{internal}$ might be the empty set. However, if $\Sigma_\textit{call}$ and $\Sigma_\textit{return}$ are empty, then the VDPA does not contain any stack-dependent behavior and can be modeled with a DFA.

We employ a simplified version of VDPAs where the stack alphabet is equal to the call alphabet, meaning that each call symbol simply pushes itself onto the stack rather than arbitrary stack symbols. This restriction does not affect the expressiveness of VDPAs since these simplified VDPAs still maintain a visible deterministic relationship between input symbols and stack operations, and since the nested structure of calls and returns remains intact, any additional information that might have been encoded through different stack symbols can instead be captured in the state space, this simplified model maintains the same computational power.

In VDPAs, the type of stack operation (push, pop, or no change) is entirely determined by the current input symbol. As the input symbol class dictates the stack behavior, the control flow of the VDPA is easier to analyze than that of a general PDA~\cite{DBLP:conf/stoc/AlurM04,DBLP:conf/icalp/AlurKMV05}.

While VDPAs are more powerful than DFAs and can recognize certain non-regular languages, they are not as expressive as general non-deterministic pushdown automata, and therefore model a strict subset of all context-free languages. For example, deterministic PDAs can recognize the language of balanced parentheses, as the stack provides a natural way to track nested structures. However, they cannot recognize all context-free languages. For example, the language of palindromes requires comparing symbols from both ends of the input, since single-stack memory cannot efficiently handle it in a deterministic manner. In addition, languages like $\{a^n b a^n \mid n \in \mathbb{N}\}$ cannot be modeled with VPDS, since $a$ would appear both in the call and return set, and by defintion, $\Sigma_{\textit{call}} \cap \Sigma_{\textit{return}} = \emptyset$.

\begin{definition}[Well-matched input sequence]
An input sequence is well-matched if it can be fully processed by a VDPA such that every call symbol in $\Sigma_\textit{call}$ has a matching return symbol in $\Sigma_\textit{return}$, and the stack is empty at the end of the sequence. An algorithm that checks whether the input sequence is well-matched can be seen in Algorithm~\ref{alg:is_balanced}.
\end{definition}

\noindent
\textbf{VDPA Example.} Figure~\ref{fig:running_example} depicts a PDA that accepts simple arithmetic expressions. Its input alphabet is partitioned as follows: $\Sigma_\textit{internal} \gets \{~1, +~\}$, $\Sigma_\textit{call} \gets \{~(~\}$, and $\Sigma_\textit{return} \gets \{~)~\}$. 

Examples of words accepted by the VDPA shown in Figure~\ref{fig:running_example} are $\{ 1, (1), 1+(1), (1) + ((1)) \}$, while words  $\{ (), )(, (1) + (), (()) \}$ are rejected by this language. Examples of input sequences that are not well matched are $)($ and $(()))$. In the first sequence $)($, the return symbol $)$ has no matching $\textit{call}$ symbol on the stack at the time of processing, that is, it popped from an empty stack, which results in an error. Likewise, when processing the sequence $(()))$, the final symbol is encountered while the stack is empty; as defined earlier, popping from an empty stack is disallowed.

\begin{figure}[t]
    \centering
    \begin{tikzpicture}[shorten >=1pt, node distance=3cm, on grid, auto] 
       \node[state, initial] (q_0)   {$s_0$}; 
       \node[state, accepting] (q_1) [right=of q_0] {$s_1$}; 
    
        \path[->] 
        (q_0) edge [loop above] node {( / push(``('')} ()
              edge  node {1} (q_1)
        (q_1) edge [loop above] node {) / pop(``('')} ()
              edge [bend left, below] node {+} (q_0);
    \end{tikzpicture}
    \caption{PDA accepting simple arithmetic expressions.}
    \label{fig:running_example}
\end{figure}
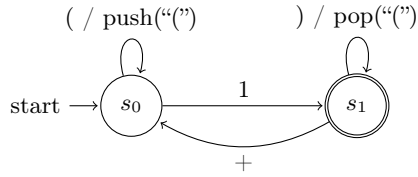

\section{Method}\label{sec:method}

PAPNI is implemented as a preprocessing step to classic RPNI implementations. As such, PAPNI can also trivially be adapted to advanced RPNI variations, such as EDSM.
PAPNI assumes prior knowledge of the decomposition of the input alphabet into push, pop, and internal symbols, and works as follows: (1) PAPNI filters out non-well-matched input sequences, and (2) PAPNI transforms well-matched input sequences to the stack-aware representation.

\subsection{Filtering of the non-well-matched sequences}
In the first step, PAPNI removes all non-well-matched input sequences from the dataset $D$. To check whether the input sequence is well-matched, we use Algorithm~\ref{alg:is_balanced}. 
The intuition behind this step is that a non-erroneous sequence generated by a VDPA must be, by definition, well-matched. Sequences that result in a non-empty stack or pop from the empty stack are not in the language of the underlying VDPA.

More specifically, these non-well-matched sequences fall into two categories: sequences with unmatched call symbols (leaving items on the stack) and sequences with unmatched return symbols (attempting to pop from an empty stack). In either case, the sequences violate the stack rules that VDPAs impose.
Including such sequences in the learning set would result in a model that would (erroneously) be able to produce sequences that are not in the language of the ground truth VDPA.  

\begin{algorithm}[t]
\caption{Function \textit{isWellMatched}}
\label{alg:is_balanced}
\begin{algorithmic}[1]
\State \textbf{Input:} $inputSeq$, $vdpaAlphabet$
\State \textbf{Output:} \textbf{True} if the sequence is well-matched, \textbf{False} otherwise
\State $counter \gets 0$
\For{$i$ in $inputSeq$}
    \If{$i \in vdpaAlphabet.callAlphabet$}
        \State $counter \gets counter + 1$
    \EndIf
    \If{$i \in vdpaAlphabet.returnAlphabet$}
        \State $counter \gets counter - 1$
    \EndIf
    \If{$counter < 0$}
        \State \Return False
    \EndIf
\EndFor
\State \Return $counter == 0$
\end{algorithmic}
\end{algorithm}

\begin{algorithm}[t]
\caption{Convert an input sequence to its stack-aware representation}
\label{alg:convert_input_seq}
\begin{algorithmic}[1]
\State \textbf{Input:} $inputSeq$, $vdpaAlphabet$
\State \textbf{Output:} $papniSequence$
\State $papniSequence \gets \{\}$
\State $stack \gets \{\}$

\For{$inputSymbol$ \textbf{in} $inputSeq$}
    \State $inputElement \gets inputSymbol$
    
    \If{$inputSymbol \in vdpaAlphabet.\textit{callAlphabet}$}
        \State $stack \gets stack \cup \{inputSymbol\}$
    \EndIf

    \If{$inputSymbol \in vdpaAlphabet.\textit{}{returnAlphabet}$}
        \State $\textit{currentTopOfStack} \gets stack.pop()$
        \State $inputElement \gets (inputSymbol, \textit{currentTopOfStack})$
    \EndIf
    
    \State $papniSequence.append(inputElement)$
\EndFor

\State \Return $papniSequence$
\end{algorithmic}
\end{algorithm}

\subsection{Transformation of input-sequences}
Well-matched input sequences are converted to the intermediate stack-aware representation, following the algorithm outlined in the Algorithm~\ref{alg:convert_input_seq}. The transformation works as follows: the algorithm processes the input sequence by pushing call symbols onto a stack and, when encountering a return symbol, pops the last call symbol from the top of the stack and pairs them together. This results in a sequence where internal and call elements appear normally, and return symbols are annotated with the corresponding call symbol they matched (which symbol was at the top of the stack when the return symbol is encountered), providing a stack-aware representation of the input.

Alur and Madhusudan~\cite{DBLP:conf/stoc/AlurM04} showed that VDPAs are, just as regular languages, closed under union, intersection, renaming, and completion. All these properties arise from the fact that the stack state is knowable and that the stack history is uniquely determined by the input sequence (which is not the case for general PDAs). This property creates a fundamental connection between the input string and stack operations, enabling us to encode stack information directly in an input string with the help of a stack-aware alphabet.

This transformation effectively "flattens" the context-free structure of visibly pushdown languages into a regular language over the stack-aware alphabet. The key insight is that once we encode which call symbol each return symbol is paired with, we no longer need an explicit stack to recognize the language. All necessary stack information is embedded directly in the transformed input.

The correctness of this transformation is derived from the correspondence between well-matched sequences and their stack histories in VDPAs. Since the stack behavior is completely determined by the input symbols, encoding this behavior directly in the input preserves all information necessary to express the behavior of the underlying VDPA. More details about correctness are outlined in the Section~\ref{sec:proof}.

\subsection{Invocation of RPNI}
Preprocessed well-matched sequences are then passed to the RPNI implementation. Execution of the RPNI algorithm over the stack-aware alphabet will return a DFA (since RPNI is used to learn regular languages) that conforms to the preprocessed data. Since the preprocessed data includes stack information, the transitions of the learned DFA encode the push and pop information used by VDPAs. 
From these steps, it can be seen that PAPNI reduces the problem of VDPA learning to the problem of DFA learning over an extended input alphabet, or as we call it in this paper, the stack-aware alphabet.


\noindent
\textbf{Implementation.}
An implementation of PAPNI can be found in \textsc{AALpy} repository~\footnote{\url{https://github.com/DES-Lab/AALpy/blob/master/aalpy/learning_algs/deterministic_passive/RPNI.py}}. As explained in this section, the implementation consists of filtering non-well-matched input sequences and converting input sequences to their ``stack-aware'' representation. The learned model is then converted to the VDPA class (since the RPNI learns a DFA). Note that in this conversion, we only change the class of the resulting model; that is, states and transitions between states are not changed. The only addition is the implementation of the stack in the VDPA class.
Both classic RPNI and EDSM are available as backend learning algorithms, and PAPNI can easily be added as a preprocessing step to other deterministic passive learning algorithms implemented in \textsc{AALpy}'s GSM framework~\cite{DBLP:conf/cav/BergA25}.

\begin{figure}[t]
    \centering
    \begin{tikzpicture}[shorten >=1pt, node distance=3cm, on grid, auto] 
       \node[state, initial] (q_0)   {$s_0$}; 
       \node[state, accepting] (q_1) [right=of q_0] {$s_1$};
       \node[state] (q_2) [right=of q_1] {$s_2$}; 
    
        \path[->] 
        (q_0) edge [loop above] node {( / push(``('')} ()
              edge  node {) / pop(``('')} (q_1)
        (q_1) edge [loop above] node {) / pop(``('')} ()
              edge  node {( / push(``('')} (q_2)
        (q_2) edge [loop above] node {( / push(``('')} ()
              edge [loop below] node {) / pop(``('')} ()
              ;
    \end{tikzpicture}
    \caption{VDPA $E$ accepting balanced parentheses. Note that state $s_2$ is the sink state, which discards all future sequences.}
    \label{fig:example_vdpa}
\end{figure}
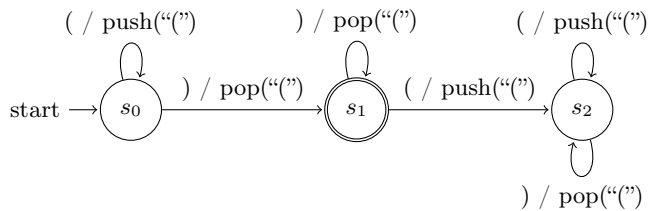

\begin{table}[t]
\centering
\begin{tabular}{|l|l|c|}
\hline
RPNI Sequences   & PAPNI Sequences                         & \multicolumn{1}{l|}{Sequence Label} \\ \hline
( , )            & (, ) pop (                              & True                                \\ \hline
(, (, ), )       & (, (, ) pop (, ) pop (                  & True                                \\ \hline
(, ), (, )       & (, ) pop (,  (, ) pop (                 & False                               \\ \hline
(, ), (, ), (, ) & (, ) pop (, (, ) pop (, (, ) pop (      & False                               \\ \hline
(, ), (, (, ), ) & (, ) pop (, ( , (, ) pop ( , ) pop (    & False                               \\ \hline
(                & -                                       & False                               \\ \hline
(, ), )          & -                                       & False                               \\ \hline
),(              & -                                       & False                               \\ \hline
(, (, )          & -                                       & False                               \\ \hline
(, (, ), ), ), ( & -                                       & False                               \\ \hline
\end{tabular}
\vspace{0.2cm}
\caption{An example of the RPNI dataset and corresponding PAPNI dataset for the simple variation of balanced parentheses. \textit{x pops y} denotes that return symbol $x$ pops call symbol $y$ from the top of the stack. Dash depicts a filtered non-well-matched sequence. The sequences conform to the automaton in Figure~\ref{fig:example_vdpa}.}
\label{tab:rpni_papni_sequences}
\end{table}

\begin{figure}[t]
    \centering
    \includegraphics[width=0.7\linewidth]{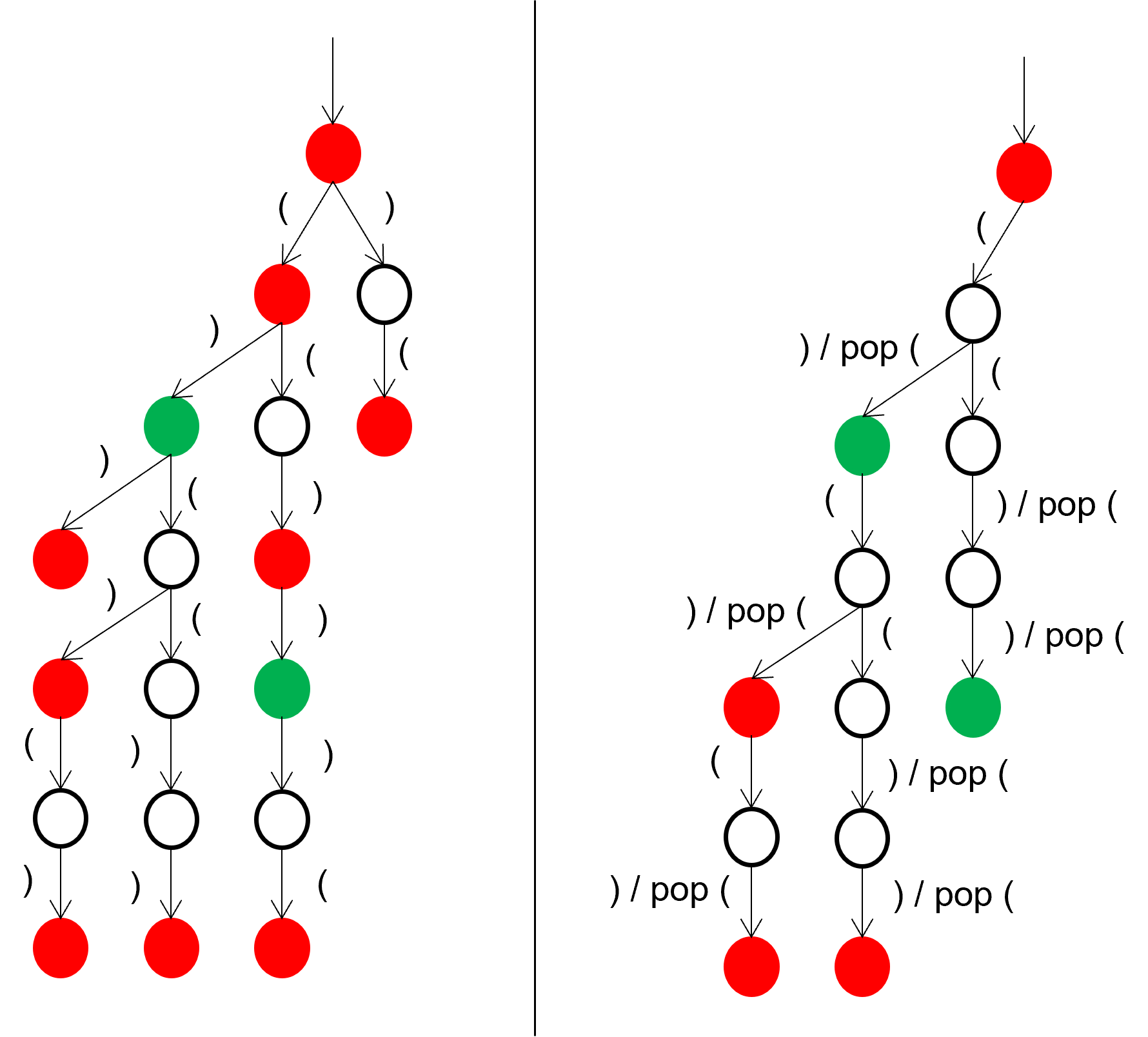}
    \caption{Prefix-tree acceptors (PTA) of RPNI (left) and PAPNI (right) with sequences from Table~\ref{tab:rpni_papni_sequences}. Green states are accepting, red states are rejecting, and the behavior of white states is not yet known. The PTA for PAPNI is noticeably smaller because it is built upon the sequences remaining after filtering.}
    \label{fig:ptas}
\end{figure}

\begin{figure}[t]
    \centering
    \begin{tikzpicture}[shorten >=1pt, node distance=3cm, on grid, auto] 
       \node[state, initial] (q_0)   {$s_0$}; 
       \node[state, accepting] (q_1) [right=of q_0] {$s_1$};
       \node[state] (q_2) [right=of q_1] {$s_2$}; 
    
        \path[->] 
        (q_0) edge [loop above] node {(} ()
              edge  node {) / pop (} (q_1)
        (q_1) edge [loop above] node {) / pop (} ()
              edge  node {(} (q_2)
        (q_2) edge [loop above] node {(} ()
              edge [loop below] node {) / pop (} ()
              ;
    \end{tikzpicture}
    \caption{DFA $E'$ constructed by RPNI.}
    \label{fig:example_rpni_result}
\end{figure}
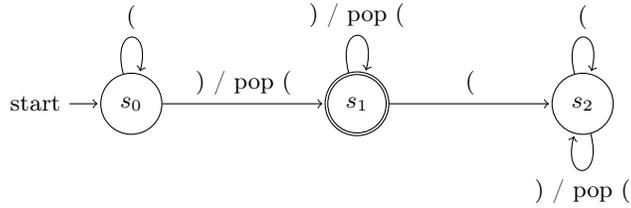

\begin{figure}[t]
    \centering
    \includegraphics[width=0.6\linewidth]{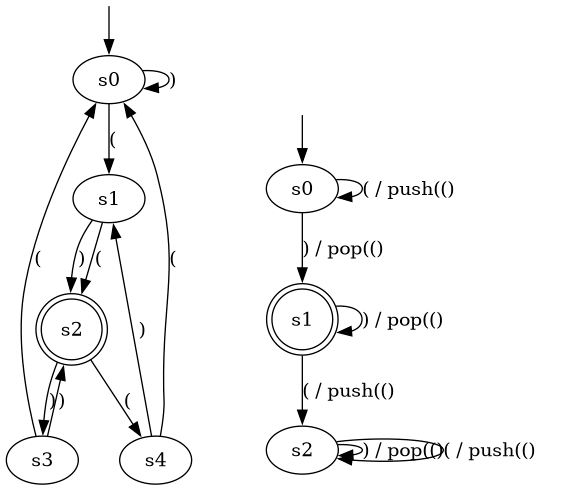}
    \caption{Models learned by RPNI (left) and PAPNI (right) with sequences from Table~\ref{tab:rpni_papni_sequences}. PAPNI's model is accurate, whereas the one by RPNI incorrectly accepts, for instance, word $)()$.}
    \label{fig:rpni_papni_models}
\end{figure}

\subsection{Example}
To illustrate the inner workings of PAPNI, we will now walk through a small example. To this end, we will attempt to infer the VDPA $E$ shown in Figure~\ref{fig:example_vdpa} using the sequences shown in Table~\ref{tab:rpni_papni_sequences}. We know that $\Sigma_{call} = \{ ~( ~\}, \Sigma_{interal} = \emptyset, \Sigma_{return} = \{ ~) ~\}$. The target language $L(E)$ is the language of balanced parentheses. The code used to produce this example is available in Appendix~\ref{apx:code}.

In the first step, PAPNI filters the given sequences to remove non-well-matched ones using Algorithm~\ref{alg:is_balanced}, which is denoted with `$-$' in Table~\ref{tab:rpni_papni_sequences}. Next, the remaining sequences are transformed into their stack-aware representation using Algorithm~\ref{alg:convert_input_seq}. This leads to the state displayed in the second column of Table~\ref{tab:rpni_papni_sequences}.

After processing the input sequences, PAPNI invokes RPNI, which builds a prefix tree acceptor (PTA)~\cite{10.5555/1830440} on the preprocessed dataset. The constructed PTA is shown on the right-hand side of Figure~\ref{fig:ptas}. If we were to use classic RPNI without data preprocessing, the PTA shown on the left-hand side of Figure~\ref{fig:ptas} would be constructed. 

The algorithm (RPNI, invoked by PAPNI) continues on the stack-aware PTA, systematically attempting to merge states if they are compatible (i.e., the merge would not lead to accepting a sequence from the data with label "False"). For more details on the inner workings of RPNI, we refer the reader to the work of De La Higuera~\cite{10.5555/1830440}. 
Eventually, the algorithm will return a DFA $E'$ over the stack-aware alphabet, which is shown in Figure~\ref{fig:example_rpni_result}. 

In the last step, PAPNI converts the DFA $E'$ to a VPDA by introducing the stack. This way, we avoid incorrectly accepting non-well-matched sequences. On the right-hand side of Figure~\ref{fig:rpni_papni_models}, we see the final model returned by PAPNI. This learned model conforms to the ground truth model used for data generation (shown in Figure~\ref{fig:example_vdpa}). 

If we were to use classic RPNI instead of PAPNI, the learning algorithm would return an incorrect model, as shown on the left-hand side of Figure~\ref{fig:rpni_papni_models}. This model incorrectly accepts, for instance, the word $)()$, which is not in the language of balanced parentheses.

\newpage

\section{On the Correctness of PAPNI}\label{sec:proof}

Before presenting the proof, we note an important practical limitation. 
Unlike standard RPNI, which learns over the given input alphabet \(\Sigma\), 
PAPNI operates over the stack-aware alphabet \(\Sigma'\). 
This alphabet is strictly larger, since for every return symbol \(r \in \Sigma_{\mathit{return}}\) 
and every call symbol \(c \in \Sigma_{\mathit{call}}\), PAPNI introduces a paired symbol \((r,c)\). 
Hence
\[
  |\Sigma'| \;=\; |\Sigma_{\mathit{internal}}|
                 + |\Sigma_{\mathit{call}}|
                 + |\Sigma_{\mathit{return}}|\cdot|\Sigma_{\mathit{call}}|,
\]
which can be substantially larger than \(|\Sigma|\). 
As a consequence, the training set \(D\) must cover a much richer alphabet, 
and in practice PAPNI may require considerably more examples than RPNI 
to uniquely identify the target language. 
This does not affect the theoretical correctness guarantees, 
but it influences sample complexity in practice.

\begin{theorem}[PAPNI Correctness]
\label{thm:papni-correctness}
Given a training set \(D\) of positive and negative examples from a visibly deterministic context‑free language \(L\), PAPNI learns a VDPA \(M\) such that \(L(M) = L\) if \(D\) contains sufficient examples to uniquely identify \(L\).
\end{theorem}

The proof of Theorem~\ref{thm:papni-correctness} consists of three main parts:
\begin{enumerate}
    \item \textit{Soundness}: Every sequence accepted by the learned model is in the target language.
    \item \textit{Completeness}: Every sequence in the target language is accepted by the learned model.
    \item \textit{Termination}: The algorithm terminates and produces a finite automaton.
\end{enumerate}

\subsection{Proof of Soundness}

\begin{lemma}[Well‑matched Filtering Preserves Language Membership]
\label{lem:wellmatched-filtering}
If \(w \in L(M)\) for some VDPA \(M\), then \(w\) is well‑matched.
\end{lemma}

\begin{proof}
By definition, a VDPA accepts a word \(w\) if and only if processing \(w\) leads to an accepting state with an empty stack. If \(w\) is not well‑matched, then either:
\begin{itemize}
    \item \(w\) contains unmatched call symbols \(\Rightarrow\) stack is non‑empty at end \(\Rightarrow\) \(w \notin L(M)\).
    \item \(w\) contains unmatched return symbols \(\Rightarrow\) pop from empty stack \(\Rightarrow\) error state \(\Rightarrow\) \(w \notin L(M)\).
\end{itemize}
Therefore, \(w \in L(M)\) implies \(w\) is well‑matched.
\end{proof}

\begin{lemma}[Stack‑aware Transformation Preserves Semantics]
\label{lem:transformation-preserves}
For any well‑matched sequence \(w\), the stack‑aware transformation \(T(w)\) uniquely encodes the stack operations performed during the execution of \(w\) on the VDPA.
\end{lemma}

\begin{proof}
Let \(w = a_1a_2\ldots a_n\) be a well‑matched sequence. During execution on VDPA \(M\), the stack content at each step is uniquely determined by the input prefix. Define
\[
  \Sigma' \;=\; \Sigma_{\mathit{internal}}
    \;\cup\;\Sigma_{\mathit{call}}
    \;\cup\;
    \{\,(\mathit{return},\mathit{call}) \mid \mathit{return}\in\Sigma_{\mathit{return}},
           \,\mathit{call}\in\Sigma_{\mathit{call}}\}.
\]
Then the transformation \(T(w)=b_1b_2\ldots b_n\) is given by
\begin{align}
    b_i = 
    \begin{cases}
        a_i, & a_i \in \Sigma_{\mathit{call}}\cup \Sigma_{\mathit{internal}},\\
        (a_i,\,c_j), & a_i\in \Sigma_{\mathit{return}}
                       \text{ and }
                       c_j\text{ is its matching call symbol.}
    \end{cases}
\end{align}
This encoding is bijective: given \(T(w)\), we can reconstruct both \(w\) and its stack behavior. Therefore, \(T(w)\) preserves all information necessary to determine acceptance.
\end{proof}

\begin{lemma}[RPNI Correctness on Stack‑aware Alphabet]
\label{lem:rpni-correctness}
RPNI correctly learns regular languages over the finite alphabet, including the stack‑aware alphabet \(\Sigma'\).
\end{lemma}

\begin{proof}
This follows directly from the correctness of RPNI as proven in~\cite{rpni,10.5555/1830440}. The fact that \(\Sigma'\) encodes stack actions does not affect RPNI’s convergence or correctness guarantees.
\end{proof}

\begin{proof}[Proof of Soundness]
Let \(M_{\textit{learned}}\) be the VDPA learned by PAPNI and let \(w\) be any sequence accepted by \(M_{\textit{learned}}\). By construction, \(w\) was classified as positive by the DFA learned by RPNI over \(\Sigma'\). By Lemma~\ref{lem:rpni-correctness}, RPNI’s classification matches the transformed training data. That data contains only well‑matched sequences from the target language (Lemma~\ref{lem:wellmatched-filtering}), and the transformation preserves membership (Lemma~\ref{lem:transformation-preserves}). Therefore, \(w \in L(M_{\textit{target}})\).
\end{proof}

\subsection{Proof of Completeness}

\begin{lemma}[Transformation Injectivity]
\label{lem:transformation-injective}
The stack‑aware transformation \(T\) is injective on well‑matched sequences.
\end{lemma}

\begin{proof}
Suppose \(T(w_1)=T(w_2)\) for well‑matched \(w_1,w_2\). The only modification in \(T\) is pairing each return symbol with its matching call. Since both sequences are well‑matched, these pairings are unique and deterministic. Hence \(w_1=w_2\).
\end{proof}

\begin{lemma}[VDPA–DFA Correspondence]
\label{lem:vdpa-dfa-correspondence}
There is a bijection between VDPAs over \(\Sigma\) and DFAs over the corresponding stack‑aware alphabet \(\Sigma'\).
\end{lemma}

\begin{proof}
\textit{Forward:} Given VDPA \(M=(Q,\Sigma,\Gamma,\delta,q_0,Z,F)\), construct DFA \(M'=(Q',\Sigma',\delta',q_0',F')\) by:
\begin{itemize}
  \item \(Q'=Q\), and $q_0' = q_0$
  \item Transitions on calls/internal symbols mimic \(\delta\) without stack effects.
  \item Transitions on \((r,c)\)\,—\,a return paired with its call\,—\,simulate a pop by design.
  \item Accepting runs in \(M'\) correspond to empty‑stack acceptance in \(M\).
\end{itemize}

\textit{Backward:} Given DFA \(M'\) over \(\Sigma'\), reconstruct VDPA \(M\) by:
\begin{itemize}
  \item Interpreting symbols in \(\Sigma_{\mathit{call}}\) as pushes.
  \item Interpreting \((r,c)\) as pops matching \(c\).
  \item Keeping internal actions unchanged.
\end{itemize}

These constructions are inverses and preserve languages w.r.t. well-matched sequences. 

However, it is important to note that the DFA does not have the capability of corretcly rejecting all non-well-matched sequences. This task is done by the stack of the VDPA, which is not available in the DFA. Since by Lemma~\ref{lem:wellmatched-filtering}, none of these words is in the language, and they will be rejected by the VDPA via the stack by definition, this does not pose a problem in the context of this proof.
\end{proof}

\begin{proof}[Proof of Completeness]
Let \(w\in L(M_{\textit{target}})\). By Lemma~\ref{lem:wellmatched-filtering}, \(w\) is well‑matched, and by Lemma~\ref{lem:transformation-preserves}, \(T(w)\) encodes its stack behavior. Lemma~\ref{lem:vdpa-dfa-correspondence} guarantees a DFA over \(\Sigma'\) that accepts \(T(w)\). If \(D\) is sufficient for RPNI to learn that DFA, then \(T(w)\) is accepted, and hence \(w\) is accepted by the learned VDPA. Thus \(L(M_{\textit{target}})\subseteq L(M_{\textit{learned}})\).
\end{proof}

\subsection{Proof of Termination}

\begin{theorem}[Termination]
\label{thm:termination}
PAPNI terminates in finite time and produces a finite automaton.
\end{theorem}

\begin{proof}
PAPNI consists of three steps:
\begin{enumerate}
  \item \textit{Filtering}: Checking well‑matchedness takes \(O(|w|)\) per sequence.
  \item \textit{Transformation}: Computing \(T(w)\) takes \(O(|w|)\) per sequence.
  \item \textit{Learning}: RPNI terminates in polynomial time~\cite{rpni}.
\end{enumerate}
The stack‑aware alphabet \(\Sigma'\) is finite:
\[
  |\Sigma'|\;\le\;|\Sigma_{\mathit{internal}}|
               +|\Sigma_{\mathit{call}}|
               +|\Sigma_{\mathit{return}}|\cdot|\Sigma_{\mathit{call}}|.
\]
Since \(D\) is finite, the entire procedure terminates and returns a finite VDPA.
\end{proof}

\section{Experimental Evaluation}\label{sec:evaluation}

We evaluate PAPNI on 13 deterministic context-free grammars. Some of these languages are also used in the literature~\cite{Yellin2021SynthesizingCG}. Those languages include variations of $X^nY^n$ languages, Dyck languages~\cite{10.5555/1096945}, variations of Dyck languages, such as languages that accept only even or odd numbers of parentheses, and a language that accepts simple arithmetic expressions. We examine the generalization capabilities of PAPNI and compare it to RPNI.

The evaluation was performed on a  Dell Latitude 5410 with an Intel Core i7-10610U processor, and 8 GB of RAM running Windows 10. We used \textsc{AALpy} v.1.5.1., an automata learning library with a mature and efficient RPNI implementation, on top of which we implemented PAPNI and added it to the library. 
Encoding and a short description of all deterministic context-free grammars under learning can be found in \textsc{AALpy}'s benchmarking utility~\footnote{\url{https://github.com/DES-Lab/AALpy/blob/master/aalpy/utils/BenchmarkVpaModels.py}}.

\subsection{Data Generation}
For all languages under test, we generated 10000 random input sequences. For each sequence, we randomly (uniform sampling) determined its length in the range [4, 50], and each input in the sequence was uniformly sampled from the input alphabet. 
Note that all positive sequences are by construction well-matched, while negative sequences are a mixture of well-matched negative sequences and not well-matched negative sequences.

Generated sequences were split into the learning and evaluation set, each containing $\sim$ 5000 sequences. We made sure that both the learning and evaluation set contained both positive and negative examples, otherwise, our evaluation metric (F1 score) would be spurious. Learning and evaluation sets are disjoint, as sequences that are found both in the learning and in the evaluation set would, by construction, be correctly evaluated (remember that RPNI, and by extension PAPNI, correctly classifies all sequences found in the learning set).

\subsection{Evaluation Metric}

In automata learning, the quality of learned models is often evaluated by examining their accuracy, that is what percentage of evaluation sequences are correctly classified. However, this accuracy metric is especially unsuitable for the evaluation of learning of deterministic context-free language learning due to the imbalance of positive and negative sample sizes. That is, a set of randomly generated input sequences will consist of negative sequences and a model that rejects all sequences might achieve high accuracy by simply rejecting all sequences. For example, if the evaluation set consists of 90\% negative sequences, a model that rejects all sequences would have an accuracy of 90\%, even though it doesn’t reflect the true behavior of the language or its ability to generate positive sequences. To avoid such an evaluation bias, we use an F1 score. 

To calculate the F1 score we considered true positives (TP), false positives (FP), false negatives (FN), and true negatives (TN). They represent the counts of correctly and incorrectly classified positive and negative sequences. Precision, recall, and F1 score are calculated as follows:

\[
\text{Precision} = \frac{\text{TP}}{\text{TP} + \text{FP}}, \quad 
\text{Recall} = \frac{\text{TP}}{\text{TP} + \text{FN}}, \quad 
\text{F1} = 2 \cdot \frac{\text{Precision} \cdot \text{Recall}}{\text{Precision} + \text{Recall}}
\]

The F1 score balances precision (how many of the predicted positives are actually correct) and recall (how many of the actual positives are successfully identified), focusing on how well the model performs with positive samples, which are often underrepresented. The higher the F1 score, the more accurate the learned model is, with an F1 score of 1.0 indicating that the model perfectly classifies all sequences in the evaluation set.

\begin{table}[t]
\centering
\caption{Comparison of RPNI and PAPNI on 13 context-free languages.}
\label{tab:results_rpni_vs_papni}
\resizebox{\textwidth}{!}{
\begin{tabular}{ccc|cccc|cccc}
 &
   &
   &
  \multicolumn{4}{c|}{RPNI} &
  \multicolumn{4}{c}{PAPNI} \\ \hline
\multicolumn{1}{c|}{Experiment} &
  \multicolumn{1}{c|}{\begin{tabular}[c]{@{}c@{}}Learning \\ Data (+/-)\end{tabular}} &
  \begin{tabular}[c]{@{}c@{}}Training \\ Data (+/-)\end{tabular} &
  \multicolumn{1}{c|}{\begin{tabular}[c]{@{}c@{}}Model\\ Size\end{tabular}} &
  \multicolumn{1}{c|}{Precision} &
  \multicolumn{1}{c|}{Recall} &
  F1 &
  \multicolumn{1}{c|}{\begin{tabular}[c]{@{}c@{}}Model\\ Size\end{tabular}} &
  \multicolumn{1}{c|}{Precision} &
  \multicolumn{1}{c|}{Recall} &
  F1 \\ \hline
\multicolumn{1}{c|}{1} &
  \multicolumn{1}{c|}{13/4988} &
  12/4987 &
  \multicolumn{1}{c|}{10} &
  \multicolumn{1}{c|}{0.571} &
  \multicolumn{1}{c|}{0.333} &
  0.421 &
  \multicolumn{1}{c|}{3} &
  \multicolumn{1}{c|}{1.0} &
  \multicolumn{1}{c|}{1.0} &
  1.0 \\ \hline
\multicolumn{1}{c|}{2} &
  \multicolumn{1}{c|}{2089/2913} &
  2087/2911 &
  \multicolumn{1}{c|}{33} &
  \multicolumn{1}{c|}{0.903} &
  \multicolumn{1}{c|}{0.928} &
  0.916 &
  \multicolumn{1}{c|}{3} &
  \multicolumn{1}{c|}{0.996} &
  \multicolumn{1}{c|}{1.0} &
  0.998 \\ \hline
\multicolumn{1}{c|}{3} &
  \multicolumn{1}{c|}{10/4992} &
  8/4990 &
  \multicolumn{1}{c|}{7} &
  \multicolumn{1}{c|}{0.75} &
  \multicolumn{1}{c|}{0.75} &
  0.75 &
  \multicolumn{1}{c|}{3} &
  \multicolumn{1}{c|}{1.0} &
  \multicolumn{1}{c|}{0.75} &
  0.85 \\ \hline
\multicolumn{1}{c|}{4} &
  \multicolumn{1}{c|}{386/4615} &
  385/4614 &
  \multicolumn{1}{c|}{53} &
  \multicolumn{1}{c|}{0.74} &
  \multicolumn{1}{c|}{0.79} &
  0.768 &
  \multicolumn{1}{c|}{2} &
  \multicolumn{1}{c|}{1.0} &
  \multicolumn{1}{c|}{1.0} &
  1.0 \\ \hline
\multicolumn{1}{c|}{5} &
  \multicolumn{1}{c|}{200/4802} &
  198/4800 &
  \multicolumn{1}{c|}{67} &
  \multicolumn{1}{c|}{0.229} &
  \multicolumn{1}{c|}{0.570} &
  0.327 &
  \multicolumn{1}{c|}{2} &
  \multicolumn{1}{c|}{1.0} &
  \multicolumn{1}{c|}{1.0} &
  1.0 \\ \hline
\multicolumn{1}{c|}{6} &
  \multicolumn{1}{c|}{159/4843} &
  157/4841 &
  \multicolumn{1}{c|}{55} &
  \multicolumn{1}{c|}{0.101} &
  \multicolumn{1}{c|}{0.369} &
  0.159 &
  \multicolumn{1}{c|}{3} &
  \multicolumn{1}{c|}{0.963} &
  \multicolumn{1}{c|}{1.0} &
  0.981 \\ \hline
\multicolumn{1}{c|}{7} &
  \multicolumn{1}{c|}{44/4957} &
  43/4956 &
  \multicolumn{1}{c|}{42} &
  \multicolumn{1}{c|}{0.101} &
  \multicolumn{1}{c|}{0.372} &
  0.159 &
  \multicolumn{1}{c|}{7} &
  \multicolumn{1}{c|}{0.693} &
  \multicolumn{1}{c|}{1.0} &
  0.819 \\ \hline
\multicolumn{1}{c|}{8} &
  \multicolumn{1}{c|}{2501/2501} &
  2499/2499 &
  \multicolumn{1}{c|}{246} &
  \multicolumn{1}{c|}{0.578} &
  \multicolumn{1}{c|}{0.615} &
  0.596 &
  \multicolumn{1}{c|}{1} &
  \multicolumn{1}{c|}{1.0} &
  \multicolumn{1}{c|}{1.0} &
  1.0 \\ \hline
\multicolumn{1}{c|}{9} &
  \multicolumn{1}{c|}{2501/2501} &
  2499/2499 &
  \multicolumn{1}{c|}{85} &
  \multicolumn{1}{c|}{0.863} &
  \multicolumn{1}{c|}{0.922} &
  0.892 &
  \multicolumn{1}{c|}{2} &
  \multicolumn{1}{c|}{1.0} &
  \multicolumn{1}{c|}{1.0} &
  1.0 \\ \hline
\multicolumn{1}{c|}{10} &
  \multicolumn{1}{c|}{2501/2501} &
  2499/2499 &
  \multicolumn{1}{c|}{63} &
  \multicolumn{1}{c|}{0.907} &
  \multicolumn{1}{c|}{0.952} &
  0.929 &
  \multicolumn{1}{c|}{2} &
  \multicolumn{1}{c|}{1.0} &
  \multicolumn{1}{c|}{1.0} &
  1.0 \\ \hline
\multicolumn{1}{c|}{11} &
  \multicolumn{1}{c|}{10/4991} &
  9/4990 &
  \multicolumn{1}{c|}{20} &
  \multicolumn{1}{c|}{0.375} &
  \multicolumn{1}{c|}{0.66} &
  0.480 &
  \multicolumn{1}{c|}{10} &
  \multicolumn{1}{c|}{0.6} &
  \multicolumn{1}{c|}{0.6} &
  0.6 \\ \hline
\multicolumn{1}{c|}{12} &
  \multicolumn{1}{c|}{222/4779} &
  221/4778 &
  \multicolumn{1}{c|}{85} &
  \multicolumn{1}{c|}{0.630} &
  \multicolumn{1}{c|}{0.656} &
  0.643 &
  \multicolumn{1}{c|}{23} &
  \multicolumn{1}{c|}{0.964} &
  \multicolumn{1}{c|}{0.977} &
  0.970 \\ \hline
\multicolumn{1}{c|}{13} &
  \multicolumn{1}{c|}{2501/2501} &
  2499/2499 &
  \multicolumn{1}{c|}{39} &
  \multicolumn{1}{c|}{0.975} &
  \multicolumn{1}{c|}{0.987} &
  0.981 &
  \multicolumn{1}{c|}{2} &
  \multicolumn{1}{c|}{1.0} &
  \multicolumn{1}{c|}{1.0} &
  1.0
\end{tabular}
}
\end{table}

\subsection{Evaluation Results} 
In Table~\ref{tab:results_rpni_vs_papni} we can see the results of RPNI and PAPNI comparison. We observed that for 7 out of 13 languages under test, PAPNI learned a model that can perfectly predict all sequences. Upon manual inspection, these models that achieved an F1 score of 1.0 fully conform to the languages used for data generation.
Other models achieve high F1 scores ($\ge$ 0.8), while for these languages RPNI failed to learn a model that meaningfully captures the language under learning. In addition, we can observe that PAPNI learners are much smaller models compared to RPNI. In cases where F1 score was 1.0, we analyzed the models and observed that they conform to the ground truth models.

Furthermore, we repeated each experiment 20 times and computed the mean and standard deviation of the F1 score for RPNI and PAPNI. The results of this evaluation can be seen in Figure~\ref{fig:f1_statistics}. We observe that the accuracy of the learned models for both techniques depends on the quality of the learning data set, and by observing the standard deviations we conclude that PAPNI is slightly more robust with respect to the learning dataset.

Finally, we would like to point out that the results of the evaluation should, and indeed do, favor PAPNI, as PAPNI is designed to learn deterministic context-free grammars, while RPNI learns deterministic regular languages, which are lower in the Chomsky hierarchy and therefore cannot completely capture the input-output behavior of higher-order languages. For any deterministic context-free language, we could find infinitely many counterexamples for the regular approximations learned by RPNI, by following the pumping lemma. Conversely, in the learning of regular languages, PAPNI achieves the same accuracy as RPNI, as the preprocessing steps are ignored (due to the lack of push/pop symbols), and PAPNI is then equivalent to its underlying algorithm, that is, RPNI. 

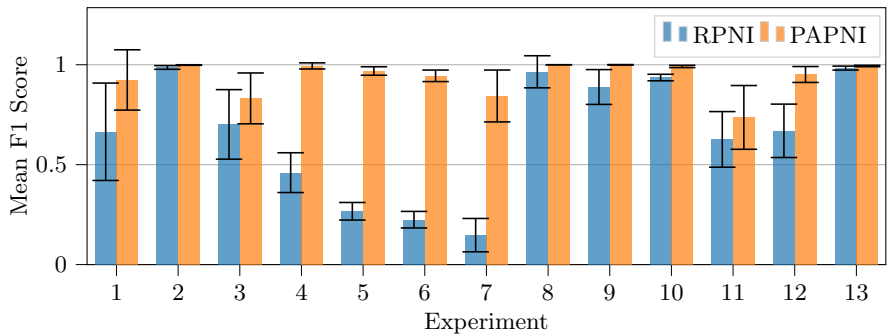
\begin{figure}[t]
    \centering
\begin{tikzpicture}

\definecolor{darkgray176}{RGB}{176,176,176}
\definecolor{darkorange25512714}{RGB}{255,127,14}
\definecolor{lightgray204}{RGB}{204,204,204}
\definecolor{steelblue31119180}{RGB}{31,119,180}

\begin{axis}[
legend cell align={left},
legend style={
  fill opacity=0.8,
  draw opacity=1,
  text opacity=1,
  at={(0.71,0.97)},
  anchor=north west,
  draw=lightgray204,
  legend columns=2,
  legend entries={RPNI,PAPNI}
},
tick align=outside,
tick pos=left,
width=\textwidth,
height=5cm,
x grid style={darkgray176},
xlabel={Experiment},
xmin=0.523, xmax=13.477,
xtick style={color=black},
y grid style={darkgray176},
ylabel={Mean F1 Score},
ymajorgrids,
ymin=0, ymax=1.28562782114469,
ytick style={color=black},
xtick={1,2,3,4,5,6,7,8,9,10,11,12,13},
ylabel near ticks,
ylabel shift={-3pt},
]
\draw[draw=none,fill=steelblue31119180,fill opacity=0.7] (axis cs:0.65,0) rectangle (axis cs:1,0.664659415101108);
\addlegendimage{ybar,ybar legend,draw=none,fill=steelblue31119180,fill opacity=0.7}
\addlegendentry{RPNI}

\draw[draw=none,fill=steelblue31119180,fill opacity=0.7] (axis cs:1.65,0) rectangle (axis cs:2,0.986327333891043);
\draw[draw=none,fill=steelblue31119180,fill opacity=0.7] (axis cs:2.65,0) rectangle (axis cs:3,0.701467967235894);
\draw[draw=none,fill=steelblue31119180,fill opacity=0.7] (axis cs:3.65,0) rectangle (axis cs:4,0.460118013429779);
\draw[draw=none,fill=steelblue31119180,fill opacity=0.7] (axis cs:4.65,0) rectangle (axis cs:5,0.266845043704471);
\draw[draw=none,fill=steelblue31119180,fill opacity=0.7] (axis cs:5.65,0) rectangle (axis cs:6,0.224730721165765);
\draw[draw=none,fill=steelblue31119180,fill opacity=0.7] (axis cs:6.65,0) rectangle (axis cs:7,0.147388018247868);
\draw[draw=none,fill=steelblue31119180,fill opacity=0.7] (axis cs:7.65,0) rectangle (axis cs:8,0.964773331337815);
\draw[draw=none,fill=steelblue31119180,fill opacity=0.7] (axis cs:8.65,0) rectangle (axis cs:9,0.888607384307213);
\draw[draw=none,fill=steelblue31119180,fill opacity=0.7] (axis cs:9.65,0) rectangle (axis cs:10,0.935971722889833);
\draw[draw=none,fill=steelblue31119180,fill opacity=0.7] (axis cs:10.65,0) rectangle (axis cs:11,0.62647540571844);
\draw[draw=none,fill=steelblue31119180,fill opacity=0.7] (axis cs:11.65,0) rectangle (axis cs:12,0.669400365910855);
\draw[draw=none,fill=steelblue31119180,fill opacity=0.7] (axis cs:12.65,0) rectangle (axis cs:13,0.983187864238041);
\draw[draw=none,fill=darkorange25512714,fill opacity=0.7] (axis cs:1,0) rectangle (axis cs:1.35,0.923998556998557);
\addlegendimage{ybar,ybar legend,draw=none,fill=darkorange25512714,fill opacity=0.7}
\addlegendentry{PAPNI}

\draw[draw=none,fill=darkorange25512714,fill opacity=0.7] (axis cs:2,0) rectangle (axis cs:2.35,0.998577088102643);
\draw[draw=none,fill=darkorange25512714,fill opacity=0.7] (axis cs:3,0) rectangle (axis cs:3.35,0.831651207638824);
\draw[draw=none,fill=darkorange25512714,fill opacity=0.7] (axis cs:4,0) rectangle (axis cs:4.35,0.994085517669056);
\draw[draw=none,fill=darkorange25512714,fill opacity=0.7] (axis cs:5,0) rectangle (axis cs:5.35,0.968746051389512);
\draw[draw=none,fill=darkorange25512714,fill opacity=0.7] (axis cs:6,0) rectangle (axis cs:6.35,0.944584267411357);
\draw[draw=none,fill=darkorange25512714,fill opacity=0.7] (axis cs:7,0) rectangle (axis cs:7.35,0.843684872080361);
\draw[draw=none,fill=darkorange25512714,fill opacity=0.7] (axis cs:8,0) rectangle (axis cs:8.35,1);
\draw[draw=none,fill=darkorange25512714,fill opacity=0.7] (axis cs:9,0) rectangle (axis cs:9.35,0.99961136122834);
\draw[draw=none,fill=darkorange25512714,fill opacity=0.7] (axis cs:10,0) rectangle (axis cs:10.35,0.99152807421987);
\draw[draw=none,fill=darkorange25512714,fill opacity=0.7] (axis cs:11,0) rectangle (axis cs:11.35,0.73653540602844);
\draw[draw=none,fill=darkorange25512714,fill opacity=0.7] (axis cs:12,0) rectangle (axis cs:12.35,0.951229208296121);
\draw[draw=none,fill=darkorange25512714,fill opacity=0.7] (axis cs:13,0) rectangle (axis cs:13.35,0.993689299098538);
\path [draw=black, semithick]
(axis cs:0.825,0.420967116711899)
--(axis cs:0.825,0.908351713490316);

\path [draw=black, semithick]
(axis cs:1.825,0.97711495718675)
--(axis cs:1.825,0.995539710595337);

\path [draw=black, semithick]
(axis cs:2.825,0.527504111526857)
--(axis cs:2.825,0.875431822944932);

\path [draw=black, semithick]
(axis cs:3.825,0.360571140742828)
--(axis cs:3.825,0.559664886116729);

\path [draw=black, semithick]
(axis cs:4.825,0.222864821022376)
--(axis cs:4.825,0.310825266386566);

\path [draw=black, semithick]
(axis cs:5.825,0.183298942869219)
--(axis cs:5.825,0.266162499462311);

\path [draw=black, semithick]
(axis cs:6.825,0.063987097583473)
--(axis cs:6.825,0.230788938912264);

\path [draw=black, semithick]
(axis cs:7.825,0.884478048352822)
--(axis cs:7.825,1.04506861432281);

\path [draw=black, semithick]
(axis cs:8.825,0.801485038519738)
--(axis cs:8.825,0.975729730094688);

\path [draw=black, semithick]
(axis cs:9.825,0.919627120588052)
--(axis cs:9.825,0.952316325191615);

\path [draw=black, semithick]
(axis cs:10.825,0.487343518151496)
--(axis cs:10.825,0.765607293285384);

\path [draw=black, semithick]
(axis cs:11.825,0.535810466767591)
--(axis cs:11.825,0.802990265054118);

\path [draw=black, semithick]
(axis cs:12.825,0.973243454283841)
--(axis cs:12.825,0.993132274192241);

\addplot [semithick, black, mark=-, mark size=5, mark options={solid}, only marks]
table {%
0.825 0.420967116711899
1.825 0.97711495718675
2.825 0.527504111526857
3.825 0.360571140742828
4.825 0.222864821022376
5.825 0.183298942869219
6.825 0.063987097583473
7.825 0.884478048352822
8.825 0.801485038519738
9.825 0.919627120588052
10.825 0.487343518151496
11.825 0.535810466767591
12.825 0.973243454283841
};
\addplot [semithick, black, mark=-, mark size=5, mark options={solid}, only marks]
table {%
0.825 0.908351713490316
1.825 0.995539710595337
2.825 0.875431822944932
3.825 0.559664886116729
4.825 0.310825266386566
5.825 0.266162499462311
6.825 0.230788938912264
7.825 1.04506861432281
8.825 0.975729730094688
9.825 0.952316325191615
10.825 0.765607293285384
11.825 0.802990265054118
12.825 0.993132274192241
};
\path [draw=black, semithick]
(axis cs:1.175,0.773118083160788)
--(axis cs:1.175,1.07487903083633);

\path [draw=black, semithick]
(axis cs:2.175,0.997233936118174)
--(axis cs:2.175,0.999920240087112);

\path [draw=black, semithick]
(axis cs:3.175,0.704408365494136)
--(axis cs:3.175,0.958894049783511);

\path [draw=black, semithick]
(axis cs:4.175,0.978729883431111)
--(axis cs:4.175,1.009441151907);

\path [draw=black, semithick]
(axis cs:5.175,0.947534706722096)
--(axis cs:5.175,0.989957396056929);

\path [draw=black, semithick]
(axis cs:6.175,0.91577778786889)
--(axis cs:6.175,0.973390746953825);

\path [draw=black, semithick]
(axis cs:7.175,0.713906606956663)
--(axis cs:7.175,0.973463137204059);

\path [draw=black, semithick]
(axis cs:8.175,1)
--(axis cs:8.175,1);

\path [draw=black, semithick]
(axis cs:9.175,0.998414722872707)
--(axis cs:9.175,1.00080799958397);

\path [draw=black, semithick]
(axis cs:10.175,0.985799587775582)
--(axis cs:10.175,0.997256560664158);

\path [draw=black, semithick]
(axis cs:11.175,0.57701935698959)
--(axis cs:11.175,0.89605145506729);

\path [draw=black, semithick]
(axis cs:12.175,0.911397748060197)
--(axis cs:12.175,0.991060668532044);

\path [draw=black, semithick]
(axis cs:13.175,0.990097302462805)
--(axis cs:13.175,0.997281295734271);

\addplot [semithick, black, mark=-, mark size=5, mark options={solid}, only marks]
table {%
1.175 0.773118083160788
2.175 0.997233936118174
3.175 0.704408365494136
4.175 0.978729883431111
5.175 0.947534706722096
6.175 0.91577778786889
7.175 0.713906606956663
8.175 1
9.175 0.998414722872707
10.175 0.985799587775582
11.175 0.57701935698959
12.175 0.911397748060197
13.175 0.990097302462805
};
\addplot [semithick, black, mark=-, mark size=5, mark options={solid}, only marks]
table {%
1.175 1.07487903083633
2.175 0.999920240087112
3.175 0.958894049783511
4.175 1.009441151907
5.175 0.989957396056929
6.175 0.973390746953825
7.175 0.973463137204059
8.175 1
9.175 1.00080799958397
10.175 0.997256560664158
11.175 0.89605145506729
12.175 0.991060668532044
13.175 0.997281295734271
};
\end{axis}

\end{tikzpicture}
    \caption{Comparison between PAPNI and RPNI F1 scores on benchmark models. All experiments were repeated 20 times to account for randomness in data generation.}
    \label{fig:f1_statistics}
\end{figure}

\noindent
\textbf{Runtime and Memory Performance.}
We examined the memory and runtime footprint of RPNI and PAPNI on one of the learned models. We generated random input sequences of increasing lengths and observed the total runtime of both algorithms, along with the total memory footprint.

Since PAPNI is a preprocessing step to standard RPNI implementation, the runtime and memory performance are largely dependent on the underlying RPNI implementation. However, since PAPNI discards non-well-matched input sequences, it has a smaller memory fingerprint and faster runtime, due to a smaller dataset. This can be observed in Figure~\ref{fig:runtime_size_comparison}. In theory, PAPNI is expected to have the same memory and runtime footprint as RPNI when provided with the same number of valid sequences. However, in practice, PAPNI discards all non-well-matched input sequences, leading to a reduced size of the training dataset, which in turn leads to a more favorable runtime and memory footprint. This behavior is reflected in the results shown in Figure~\ref{fig:runtime_size_comparison}.

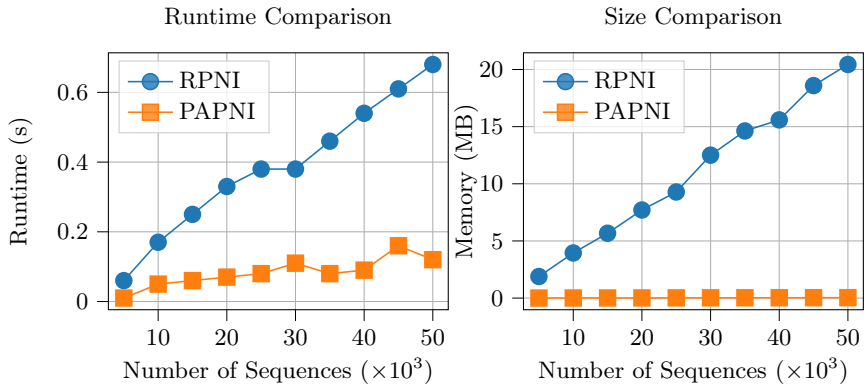
\begin{figure}[t]
    \centering
\begin{tikzpicture}

\definecolor{darkgray176}{RGB}{176,176,176}
\definecolor{darkorange25512714}{RGB}{255,127,14}
\definecolor{lightgray204}{RGB}{204,204,204}
\definecolor{steelblue31119180}{RGB}{31,119,180}

\begin{groupplot}[group style={group size=2 by 1}]
\nextgroupplot[
legend cell align={left},
legend style={
  fill opacity=0.8,
  draw opacity=1,
  text opacity=1,
  at={(0.03,0.97)},
  anchor=north west,
  draw=lightgray204
},
tick align=outside,
tick pos=left,
title={Runtime Comparison},
x grid style={darkgray176},
xlabel={Number of Sequences (\(\times 10^3\))},
xmajorgrids,
width=\textwidth / 2,
height=5cm,
xmin=2750, xmax=52250,
xtick style={color=black},
y grid style={darkgray176},
ylabel={Runtime (s)},
ymajorgrids,
ymin=-0.0235, ymax=0.7135,
ytick style={color=black},
ylabel near ticks,
ylabel shift={3pt},
xtick={10000,20000,30000,40000,50000}, 
xticklabels={10,20,30,40,50}, 
scaled x ticks=false
]
\addplot [semithick, steelblue31119180, mark=*, mark size=3, mark options={solid}]
table {%
5000 0.06
10000 0.17
15000 0.25
20000 0.33
25000 0.38
30000 0.38
35000 0.46
40000 0.54
45000 0.61
50000 0.68
};
\addlegendentry{RPNI}
\addplot [semithick, darkorange25512714, mark=square*, mark size=3, mark options={solid}]
table {%
5000 0.01
10000 0.05
15000 0.06
20000 0.07
25000 0.08
30000 0.11
35000 0.08
40000 0.09
45000 0.16
50000 0.12
};
\addlegendentry{PAPNI}

\nextgroupplot[
legend cell align={left},
legend style={
  fill opacity=0.8,
  draw opacity=1,
  text opacity=1,
  at={(0.03,0.97)},
  anchor=north west,
  draw=lightgray204
},
tick align=outside,
tick pos=left,
title={Size Comparison},
x grid style={darkgray176},
xlabel={Number of Sequences (\(\times 10^3\))},
xmajorgrids,
xmin=2750, xmax=52250,
xtick style={color=black},
y grid style={darkgray176},
width=\textwidth / 2,
height=5cm,
ylabel={Memory (MB)},
ymajorgrids,
ymin=-1.01835541725159, ymax=21.4614353656769,
ytick style={color=black},
ylabel near ticks,
ylabel shift={-6pt},
xtick={10000,20000,30000,40000,50000}, 
xticklabels={10,20,30,40,50}, 
scaled x ticks=false
]
\addplot [semithick, steelblue31119180, mark=*, mark size=3, mark options={solid}]
table {%
5000 1.8873348236084
10000 3.94776725769043
15000 5.67314720153809
20000 7.70704460144043
25000 9.28195762634277
30000 12.5037670135498
35000 14.6226177215576
40000 15.5918788909912
45000 18.5895900726318
50000 20.4396266937256
};
\addlegendentry{RPNI}
\addplot [semithick, darkorange25512714, mark=square*, mark size=3, mark options={solid}]
table {%
5000 0.00345325469970703
10000 0.00722322592972917
15000 0.0103801519427197
20000 0.0141015720463907
25000 0.0169831888834476
30000 0.0228781304003225
35000 0.0267549894895944
40000 0.0285284457130062
45000 0.0340133549600962
50000 0.037398365174769
};
\addlegendentry{PAPNI}
\end{groupplot}

\end{tikzpicture}
    \caption{Comparison of PAPNI's and RPNI's runtime and memory consumption.}
    \label{fig:runtime_size_comparison}
\end{figure}

\section{Discussion on possible limitations and their mitigation}\label{sec:limitations}

\noindent
\textbf{Data Availability.}
This limitation is assumed for all passive learning algorithms, like RPNI or \textsc{Alergia}. However, we point out that this limitation might be more prominent for PAPNI, since the algorithm infers a model based only on positive and negative \textit{well-matched} sequences.
Depending on the structure of the language under learning this limitation might be negligible, while for others it could severely limit the quality of the learned model.

Consider, for example, learning an XML or JSON from positive and negative examples. Intuitively, one might assume that we will find many well-matched positive sequences (valid XML and JSON files), while well-matched negative sequences are not so readily available. We might, for example, attempt to mitigate this issue by mutating the original positive examples and checking (with the syntax verifier) if the resulting example is still valid.

In addition, as mentioned at the beginning of Section~\ref{sec:proof}, PAPNI operates over an extended input alphabet that is larger than RPNI's input alphabet (due to all possible combinations of call and return symbols). Therefore, PAPNI requires a larger set of well-matched input sequences to converge to the ground truth, as the provided dataset should ideally cover all call-return pairs.   

\textbf{Absence of visibly-deterministic alphabet.}
As previously explained, PAPNI assumes the apriori knowledge of the decomposition of the input alphabet into push, pop, and internal symbols. We postulate that this limitation is minimal since in many real-world grammars, push and pop symbols are clearly defined (such as parentheses, square brackets, HTML tags, etc.). Alternatively, if such a decomposition of the input alphabet is not known, one might develop a heuristic that examines all positive and negative sequences and observes patterns that might indicate which symbols affect the stack or not.

\section{Conclusion}\label{sec:conclusion}

We introduce PAPNI, a passive learning algorithm that extends the classic RPNI implementations to visibly deterministic pushdown automata. PAPNI functions as a preprocessing step to RPNI: it starts by filtering out all non-well-matched input sequences. The remaining well-matched sequences are then transformed into an equivalent stack-aware representation, which is then passed to RPNI. As such, PAPNI reframes the challenge of learning deterministic context-free grammars into regular language learning over a stack-aware input alphabet.

Experimental evaluation on 13 benchmark languages demonstrates that PAPNI significantly outperforms RPNI on all context-free grammars under consideration, achieving perfect classification (F1 = 1.0) on 7 out of 13 test languages while maintaining small model sizes. Moreover, the preprocessing strategy used by PAPNI can be adapted to other passive learning algorithms, such as EDSM, suggesting a broader framework for extending regular grammar inference methods to context-free settings.

From a practical perspective, PAPNI could be used in protocol analysis and structured data processing. The algorithm's ability to learn XML, JSON, and programming language grammars from examples makes it practical for reverse engineering legacy systems and validating parsers for domain-specific languages. In recent years, we have also observed interest in applying automata learning to the modeling of recurrent neural networks (RNNs). The majority of such research~\cite{DBLP:conf/icgi/EyraudLJGCHS23,DBLP:conf/ifm/MuskardinAPT22,DBLP:conf/cdmake/MayrY18,DBLP:journals/ml/WeissGY24} focused on inference of DFAs from RNNs. PAPNI could be used to analyze RNN decision-making processes on more expressive context-free grammars, which in turn might reveal additional insights about RNN internal behavior.

\noindent
\textbf{Acknowledgements.} The research reported in this paper has been partly funded by the European Union’s Horizon 2020 research and innovation program within the framework of Chips Joint Undertaking (Grant No. 101112268). This work has been supported by Silicon Austria Labs (SAL), owned by the Republic of Austria, the Styrian Business Promotion Agency (SFG), the federal state of Carinthia, the Upper Austrian Research (UAR), and the Austrian Association for the Electric and Electronics Industry (FEEI).

\bibliographystyle{plain}
\bibliography{bibliography}

\newpage
\appendix
\section{Code Example}\label{apx:code}
The following code was used to produce the running example in Section~\ref{sec:method}. In order to make it run properly, a Python 3 installation is required. Furthermore, the \texttt{aalpy} library must be available, which can be installed, for example, using the command \texttt{pip install aalpy}. 

\begin{lstlisting}
from aalpy import run_RPNI, run_PAPNI
from aalpy.automata import VpaAlphabet

dataset = [
    ('', False),
    ('()', True),
    ('(())', True),
    ('()()', False),
    ('()()()', False),
    ('()(())', False),
    ('(', False),
    ('())', False),
    (')(', False),
    ('(()', False),
    ('(()))(', False),
]

vpa_alphabet = VpaAlphabet(
    internal_alphabet=[],
    call_alphabet=['('],
    return_alphabet=[')']
)

rpni_dfa = run_RPNI(dataset, automaton_type='dfa')
# RPNI model size: 5
print(f'RPNI model size: {rpni_dfa.size}')

papni_model = run_PAPNI(dataset, vpa_alphabet, 
                        algorithm='gsm')
                        
# PAPNI model size: 3
print(f'PAPNI model size: {papni_model.size}')

# For the visualizations of th learned models
# rpni_dfa.visualize(path='./rpni-model')
# papni_model.visualize(path='./papni-model')
\end{lstlisting}

\end{document}